\documentclass[11pt]{article}

\usepackage[left=0.75in,right=0.75in,top=1in,bottom=1in]{geometry}
\usepackage[utf8]{inputenc}
\usepackage[T1]{fontenc}
\usepackage{kpfonts}
\usepackage[usenames,dvipsnames]{color}
\usepackage[colorlinks=true,citecolor=blue,linkcolor=blue]{hyperref}
\usepackage[capitalize,nameinlink]{cleveref}
\usepackage{amsmath,amsthm,amsfonts,amssymb,mathrsfs,paralist,mathtools}
\usepackage{cases, bm, enumitem}

\newcommand\numberthis{\addtocounter{equation}{1}\tag{\theequation}}
\newtheorem{lemma}{Lemma}
\newtheorem{theorem}[lemma]{Theorem}
\theoremstyle{definition}
\newtheorem{definition}[lemma]{Definition}
\newtheorem{corollary}[lemma]{Corollary}
\newtheorem{proposition}{Proposition}

\newcommand\floor[1]{\lfloor{#1}\rfloor}
\newcommand\ceil[1]{\lceil{#1}\rceil}

\usepackage{subfiles}
\usepackage{cite}

\begin{document}

\title{Quantum Bicyclic Hyperbolic Codes}
\author{
Sankara Sai Chaithanya Rayudu\\
Department of Electrical Engineering\\
Indian Institute of Technology Madras\\
Chennai 600036 India\\
Email: ee14b050@ee.iitm.ac.in
\and
Pradeep Kiran Sarvepalli\\
Department of Electrical Engineering\\
Indian Institute of Technology Madras\\
Chennai 600036 India\\
Email: pradeep@ee.iitm.ac.in
}
\date{\vspace{-5ex}}
\maketitle

\begin{abstract}
Bicyclic codes are a generalization of the one dimensional (1D) cyclic codes to two dimensions (2D). Similar to the 1D case, in some cases, 2D cyclic codes can also be constructed to guarantee a specified minimum distance. Many aspects of these codes are yet unexplored. Motivated by the problem of constructing quantum codes, in this paper, we study some structural properties of certain bicyclic codes. We show that a primitive narrow-sense bicyclic hyperbolic code of length $n^2$ contains its dual if and only if its design distance is lower than $n-\Delta$, where $\Delta=\mathcal{O}(\sqrt{n})$. We extend the sufficiency condition to the non-primitive case as well. We also show that over quadratic extension fields, a primitive bicyclic hyperbolic code of length $n^2$ contains Hermitian dual if and only if its design distance is lower than  $n-\Delta_h$, where $\Delta_h=\mathcal{O}(\sqrt{n})$. Our results are analogous to some structural results known for BCH and Reed-Solomon codes. They further our understanding of bicyclic codes. We also give an application of these results by showing that we can construct two classes of quantum bicyclic codes based on our results.
\end{abstract}

\section{Introduction}
Cyclic codes 
are an important class of error-correcting codes. 
Many popular codes, such as BCH codes and Reed-Solomon codes, are cyclic codes. 
Cyclic codes with guarantees on the minimum distance of the code are easy to construct. 
Many subclasses of cyclic codes also have efficient decoders making them suitable for practical applications.  
For quantum error correction, a classical code can be used to construct quantum code \cite{steane96, calderbank96,calderbank98,ashikhmin01} if the code contains its (Euclidean or Hermitian) dual.
Using these constructions many 
(cyclic) quantum codes have been proposed \cite{grassl97,grassl04,ketkar06}. 
Grassl {\em et al.} gave a simple test for identifying cyclic codes that contain their duals \cite{grassl97}.
Steane \cite{steane99}  gave a condition, based on the designed distance, to check whether a primitive binary BCH contains its Euclidean dual. 
Subsequently, Aly {\em et al.} \cite{aly07} extended this result to the higher alphabet and non-primitive codes. 
They proved that a primitive BCH code of length $n$ contains its dual when its design distance is less than $\delta = \mathcal{O}(\sqrt{n})$.
These results are based on significant structural results of cyclic codes.
However, most of the previous work has been limited to one dimension \cite{yuan17, liu17, zhang18, li13, chen15, liu19}, even though classically, cyclic codes have been generalized to higher dimensions.

Two dimensional (2D) cyclic codes, also called bicyclic codes,  are a generalization of one-dimensional cyclic codes to two dimensions. 
In the case of one-dimensional cyclic codes, codewords can be considered as vectors. 
The codewords of bicyclic codes can be viewed as matrices. 
Imai \cite{IMAI19771} introduced a general theory of 2D cyclic codes.
Since then, there has been extensive work on bicyclic codes, see \cite{blahut08} for a good overview of related work and references.
However, there does not appear to be much work on quantum bicyclic codes. 
 
There are many significant differences between cyclic and bicyclic codes which makes the analysis of bicyclic codes much more challenging than cyclic codes. 
The characterization of 2D cyclic codes is somewhat more complex than the cyclic codes. 
For instance, bicyclic codes do not have a unique generator polynomial, unlike the cyclic codes. 
Furthermore, the division of polynomials by bivariate polynomials does not lead to unique remainders.
All these reasons motivate our study of quantum bicyclic codes. 
In this paper, we focus on a class of bicyclic codes called hyperbolic codes, see \cite{blahut08}. 
(Note that the term hyperbolic codes also refers to a class of topological quantum codes.)
 
Our main contributions are as follows:
\begin{compactenum}[i)]

\item We give necessary and sufficient condition for a bicyclic hyperbolic code to contain its Euclidean dual.
\item For a bicyclic hyperbolic code, defined over a quadratic extension field, we also give a necessary and sufficient condition for it to contain its Hermitian dual. 
\item We construct new quantum bicyclic codes.
\end{compactenum}

Our analysis of the cyclotomic cosets of bicyclic codes could be of independent interest and use in the study of cyclic codes over higher dimensions. Our paper is organized as follows. After a brief review of the necessary background in Section~\ref{sec:bg}, we prove structural results on bicyclic codes in 
Sections~\ref{sec:main-ed}~and~\ref{sec:main-hd}. 
We conclude with a brief discussion on the significance of these results and the scope for future work. 
\section{Background}\label{sec:bg}
Let $q$ be a prime power, and $\mathbb{F}_q$ denote a finite field of $q$ elements. 
We denote by $\mathbb{Z}_m$ the set of integers $\{0,1,\ldots, m-1\}$.

\subsection{Bicyclic codes}
A $[n,k]_q$ linear code is a $k$-dimensional subspace of $\mathbb{F}_q^n$.
A bicyclic code of length $n_1\times n_2$ can be viewed as a subspace of $\mathbb{F}_q^{n_1\times n_2}$.
We give a quick review of bicyclic codes; readers interested in more details should refer to \cite{blahut08, IMAI19771}.

\begin{definition}[Bicyclic codes]
Let $\gcd(n_1n_2,q)=1$ and  $C$ a linear code of length $n_1n_2$ over a field $\mathbb{F}_q$, whose codewords can be written as two dimensional arrays of size $n_1 \times n_2$. If $C$ is closed under both circular right shift of columns and circular down shift of rows, then $C$ is called a 
bicyclic code of length $n_1\times n_2$ over $\mathbb{F}_q$. 
\end{definition}

\noindent Let  $c$ be a codeword of a bicyclic code $C$ of size $n_1\times n_2$. 
{
Then we can write $c$ as follows: 
    \begin{equation*}
        c =
        \begin{bmatrix}
            c_{0,0} & c_{0,1} & c_{0,2} & \dots  & c_{0,n_2-1} \\
            c_{1,0} & c_{1,1} & c_{1,2} & \dots  & c_{1,n_2-1} \\
            \vdots & \vdots & \vdots & \ddots & \vdots \\
            c_{n_1-1,0} & c_{n_1-1,1} & c_{n_1-1,2} & \dots  & c_{n_1-1,n_2-1}
        \end{bmatrix}
    \end{equation*}
}
We denote by  $c^{(0,1)}$ the codeword obtained by circular right shift of columns of $c$  and $c^{(1,0)}$ the codeword obtained by circular down shift of rows of $c$.

\begin{minipage}{0.45\textwidth}
    \begin{equation*}
        c^{(0,1)} = 
        \begin{bmatrix}
            c_{0,n_2-1} & c_{0,0} & c_{0,1} & \dots  & c_{0,n_2-2} \\
            c_{1,n_2-1} & c_{1,0} & c_{1,1} & \dots  & c_{1,n_2-2} \\
            \vdots & \vdots & \vdots & \ddots & \vdots \\
            c_{n_1-1,n_2-1} & c_{n_1-1,0} & c_{n_1-1,1} & \dots  & c_{n_1-1,n_2-2}
        \end{bmatrix}
    \end{equation*}
\end{minipage}
\begin{minipage}{0.55\textwidth}
    \begin{equation*}
        c^{(1,0)} =
        \begin{bmatrix}
            c_{n_1-1,0} & c_{n_1-1,1} & c_{n_1-1,2} & \dots  & c_{n_1-1,n_2-1} \\
            c_{0,0} & c_{0,1} & c_{0,2} & \dots  & c_{0,n_2-1} \\
            \vdots & \vdots & \vdots & \ddots & \vdots \\
            c_{n_1-2,0} & c_{n_1-2,1} & c_{n_1-2,2} & \dots  & c_{n_1-2,n_2-1}
        \end{bmatrix}
    \end{equation*}
\end{minipage}

\subsubsection*{Polynomial representation}
As in the case of 1D cyclic codes, it is convenient to represent codewords as polynomials over $\mathbb{F}_q$. 
Let $\mathbb{F}_q[X, Y]$ denote the ring of polynomials in two variables over $\mathbb{F}_q$.
Codewords of a bicyclic code can be represented as bivariate polynomials in $\mathbb{F}_q[X, Y]$.
We will call the polynomial associated with a codeword as a code polynomial. The code polynomials of a bicyclic code of size $n_1\times n_2$ will be of the following form.
\begin{equation*}
        c(X,Y) = \sum_{i=0}^{n_1-1} \sum_{j=0}^{n_2-1} c_{i,j}\  X^i Y^j
\end{equation*}
Then the code polynomials corresponding to codewords $c^{(1,0)}$ and $c^{(0,1)}$ will be
\begin{align*}
    c^{(1,0)}(X,Y) &= X c(X,Y) \text{ mod } (X^{n_1}-1, Y^{n_2}-1) \\
    c^{(0,1)}(X,Y) &= Y c(X,Y) \text{ mod } (X^{n_1}-1, Y^{n_2}-1)
\end{align*}

Linearity of bicyclic codes implies that for any code polynomial $c(X,Y)$ in a bicyclic code $C$, and $a(X,Y)\in \mathbb{F}_q[X,Y]$ we also have
$a(X,Y) c(X,Y)\bmod (X^{n_1}-1, Y^{n_2}-1)$ in $C$.
Therefore, bicyclic codes can be viewed as ideals in the quotient  polynomial ring $R := \mathbb{F}[X,Y]/\langle X^{n_1}-1, Y^{n_2}-1\rangle$.

\subsubsection*{Characterization}
It is well know that when $\gcd(n,q)=1$, a 1D cyclic code of length $n$ over $\mathbb{F}_q$ is completely characterized by a unique monic generator polynomial $g(X)$
where $g(X)$ divides $X^n-1$. 
Due to lack of direct generalization of polynomial division from 1D to 2D, We cannot always characterize bicyclic codes by a unique generator polynomial.
Therefore, bicyclic codes are better analyzed in terms of the common zeros of the code polynomials.

Given a bicyclic code of length $n_1\times n_2$ over $\mathbb{F}_q$, the common zeros of all the codeword polynomials, $X^{n_1}-1$ and $Y^{n_2}-1$ completely characterize the bicyclic code when $q$ is co-prime to both $n_1$ and $n_2$. Any polynomial of length $n_1\times n_2$ which vanishes at these common zeros is a codeword of that bicyclic code. Common zeros will be of the form $(\alpha^i,\beta^j)$ where $\alpha$ and $\beta$ are the $n_1^{th}$ and $n_2^{th}$ primitive roots of unity.
The set of all such possible zeros is
\begin{equation} \label{}
    \Omega = \left\{ (\alpha^i,\beta^j)\ \middle|\ 0\leq i < n_1, 0\leq j < n_2 \right\}.
\end{equation}
It is easy to keep track of zeros just by the exponents of $\alpha$ and $\beta$. 
Therefore the {\em defining set} of a bicyclic code $C$ is given as follows: 
\begin{equation}
    Z = \left\{ (x,y)\ \middle|\ (\alpha^x,\beta^y) \text{ is a common zero of  }C \right\}
\end{equation}
Observe that since the codewords are over the field $\mathbb{F}_q$ and $\alpha$ and $\beta$ may lie in an extended field of $\mathbb{F}_q$, If $(x,y)$ is in the defining set of code $C$ then all the points of the form $(xq^k,yq^k)$ for $k \in \mathbb{Z},k\geq 0$, should also be in the defining set of $C$. 
The set of all such points is called the $q$-ary cyclotomic coset of $(x,y)$ modulo $n_1$ and $n_2$ and denoted as $Coset_{(x,y)}$.
\begin{equation}
Coset_{(x,y)} = \left\{ (xq^k \bmod n_1, yq^k \bmod  n_2)\ \middle|\ k \in \mathbb{Z}, k \geq 0 \right\}
\end{equation}

\subsection{Bicyclic hyperbolic codes}
\begin{definition}[Bicyclic hyperbolic codes]
Let $n_i=q^{m_i}-1$ and $\gcd(n_i,q)=1$ for $i=1,2$.
A bicyclic code $C$ of length $n_1 \times n_2$ over $\mathbb{F}_q$ is called a hyperbolic code with designed distance $\delta$ if the defining set of $C$ is of the following form
\begin{equation}
    Z = \bigcup_{(x,y) \in Z_{des}} \hspace{-10pt} Coset_{(x,y)} \label{eq:bch-Z}
\end{equation}
\begin{equation}
    Z_{des} = \left\{ ((a+x') \bmod  n_1, (b+y') \bmod  n_2)\ \middle|\ x'y' < \delta\right\}\label{eq:hyper-zdes}
\end{equation}
where $1 \leq x' \leq n_1$, $1 \leq y' \leq n_2$. $Z_{des}$ is the designed set of $C$ which completely characterizes the code.
\end{definition}
In the above definition, there is a freedom to choose the values of  $a$ and $b$. A bicyclic hyperbolic code with designed distance $\delta$ is guaranteed to have a minimum distance greater than or equal to $\delta$.
Our definition of bicyclic hyperbolic code is slightly different from the definition given in \cite{blahut08}. Analogous to one dimensional case we say the hyperbolic code is primitive if $n_1 = n_2 = n = q^m-1$ 
and narrow-sense if $a=b=0$.
In this paper we focus on the narrow-sense primitive bicyclic hyperbolic codes with $n_1=n_2=n=q^m-1$.
In the rest of this paper, We denote a bicyclic hyperbolic code of design distance $d$ by $\mathscr{H}(n\times n, q;d)$.
\section{Euclidean Dual Containing Bicyclic Codes}\label{sec:main-ed}
The Euclidean inner product of two elements $u,v$ in $\mathbb{F}_q^{n_1\times n_2}$ is defined as follows:
\begin{align}
u\cdot v = \sum_{i,j} u_{i,j}v_{i,j}. \label{eq:ip}
\end{align}
The Euclidean dual code $C^{\perp}$ of a  linear bicyclic code $C \subseteq \mathbb{F}_q^{n_1\times n_2}$ is defined as 
\begin{align}
C^{\perp}=\left\{ u\in  \mathbb{F}_q^{n_1\times n_2}\ \middle|\ u\cdot v = 0\  \text{ for all }  v \in C\right\}.\label{eq:e-dual}
\end{align} 
If $C$ is bicyclic, then $C^\perp$ is also bicyclic and its zeros can be given in terms of the zeros of $C$.
When $gcd(n_i,q)=1$, the defining set of $C^\perp$ is completely characterized in terms of the defining set of $C$ \cite{IMAI19771}.
Suppose $Z$ is the defining set of $C$ and $Z^\perp$ the defining set of $C^\perp$. Then 
\begin{align}
Z^\perp=Z_{tot}\setminus Z^{-1}, \label{eq:z-dual}
\end{align}
where
\begin{equation}
Z_{tot} =\left\{(x,y)\ \middle|\ x\in \mathbb{Z}_{n_1}, y\in \mathbb{Z}_{n_2}\right\} \\
\end{equation}
\begin{equation} 
Z^{-1} = \left\{(-x \bmod n_1, -y\bmod n_2)\ \middle|\ (x,y)\in Z\right\}.
\end{equation}
A simple test to check if a  cyclic code  contains its dual in terms of its defining set was given in 
\cite{grassl97}. 
The same condition holds for bicyclic codes, which we are giving in the following lemma.  
Although straightforward, we give the proof for completeness.

\begin{lemma}\label{lm:dual-containing}
Let $C$ be the bicyclic code of length $n_1\times n_2$ over $\mathbb{F}_q$ such that $\text{gcd}(n_i,q) = 1$ and  let $Z$ be the defining set of $C$. Then $C$ contains its Euclidean dual $C^\perp$ if and only if 
\begin{equation}\label{lm:dual-containing_eq_1}
    Z\cap Z^{-1} = \emptyset
\end{equation}
where $Z^{-1} = \left\{ (-x \bmod n_1,-y \bmod n_2)\ \middle|\ (x,y)\in Z\right\}$.
\end{lemma}
\begin{proof}
Let $C_1$ and $C_2$ be two bicyclic cyclic codes with defining sets $Z_1$ and $Z_2$. The code $C_1$ is contained in $C_2$ means that all the codewords that are in $C_1$ are also in $C_2$. This is equivalent to saying that the defining set (or common zero set) of $C_2$ is contained in that of $C_1$
i.e., $Z_2\subseteq Z_1$.
 From Eq.~\eqref{eq:z-dual}, $Z_{tot}\setminus Z^{-1}$ is the defining set of dual code $C^{\perp}$. 
Hence, $C^{\perp} \subseteq C$ if and only if $Z \subseteq Z_{tot}\setminus Z^{-1}$.
This is possible if and only if $Z \cap Z^{-1} = \emptyset$.
\end{proof}

The above condition for Euclidean dual containing can be further simplified for bicyclic hyperbolic codes with a designed set $Z_{des}$ since the  defining set $Z$ is completely characterized $Z_{des}$.
\begin{lemma}\label{lemma_2}
Let $C$ be a bicyclic  code of length $n_1\times n_2$ over $\mathbb{F}_q$ such that $\text{gcd}(n_1n_2,q) = 1$. Let $Z_{des}$ and $Z$ be the designed set and defining set of $C$ respectively. Then $C^\perp \subseteq C$ if and only if 
\begin{equation}\label{lemma_2_eq_1}
    Z_{des}\cap Z^{-1} = \emptyset
\end{equation}
where $Z^{-1} = \left\{ (-x \bmod n_1, -y \bmod n_2)\ \middle|\ (x,y)\in Z\right\}$.
\end{lemma}

\begin{proof}
By Lemma~\ref{lm:dual-containing}, $C\supseteq C^\perp$ if and only if $Z\cap Z^{-1}=\emptyset$.
To prove the lemma it suffices to show that $Z\cap Z^{-1}=\emptyset$ if and only  $Z_{des}\cap Z^{-1} = \emptyset$.
Since $Z_{des}\subseteq Z$ if $Z\cap Z^{-1}=\emptyset$, then Eq.~\eqref{lemma_2_eq_1} holds.
Suppose now that Eq.~\eqref{lemma_2_eq_1} holds.
If $(x,y)\in Z_{des}$, then $(x,y) \not\in Z^{-1}$.
As $Z^{-1}$ is the union of cyclotomic cosets, it follows that 
 $C_{(x,y)}\cap Z^{-1}=\emptyset$. 
From Eq.~\eqref{eq:bch-Z}, we have
 $Z=\cup_{(x,y)\in Z_{des}}Coset_{(x,y)}$, it follows that $Z\cap Z^{-1}=\emptyset$.
\end{proof}



The following theorem gives an easy test, based on the designed distance, to check if a primitive narrow-sense hyperbolic code contains its Euclidean dual. In the proof of the following theorem \ref{theorem_3} we will use lemma \ref{apx_lemma_1}, which is stated after the theorem \ref{theorem_3} to understand the motivation for lemma \ref{apx_lemma_1}.

\begin{theorem}[Euclidean dual containing bicyclic codes]\label{theorem_3}
A primitive narrow-sense hyperbolic code of length $n\times n$ over $\mathbb{F}_q$, where $n=q^m-1$ and $m>3$, contains its Euclidean dual if and only if the design distance $d$ satisfies $d \leq \delta$, where
\begin{equation} \label{theorem_3_eq_1} 
        \delta = 
    \begin{cases}
        (q^m-1) - 2(q^{m/2} - 1) & m \text{ is even} \\
        (q^m-1) - (q^{\frac{m-1}{2}}) & m \text{ is odd}
    \end{cases}
\end{equation}
\end{theorem}
\begin{proof}
    Let $C= \mathscr{H}(n\times n, q;d)$.
    First, we show the sufficiency of $d\leq \delta$ for $C$ to be dual containing. 
    For proving $C \supseteq C^\perp$ when $d \leq \delta$, it is  enough to show that
    $D=\mathscr{H}(n\times n, q;\delta)$ contains $D^\perp$. 
    Since $D\subseteq C$, this implies $C^\perp \subseteq C$.
    The designed set for $\mathscr{H}(n \times n,q;\delta)$ is given below: 
    \begin{equation}\label{theorem_3_eq_2} 
        Z_{des} = \left\{ (x,y)\ \middle|\ xy <\delta,\,1 \leq x,y \leq n\right\}
    \end{equation}
    By Lemma~\ref{lemma_2}, we need to show Eq.~\eqref{lemma_2_eq_1} holds i.e., $Z_{des}\cap Z^{-1} = \emptyset$ for $\mathscr{H}(n\times n,q;\delta)$.
    By the definition of $Z_{des}$ and $Z^{-1}$, for every $(x,y) \in Z_{des}$ the corresponding points $(-xq^l\mod{n}, -yq^l\mod{n})$ must be in $Z^{-1}$ for all $l \in \mathbb{Z}_m$. Therefore $Z_{des}\cap Z^{-1} = \emptyset$ if and only if the following inequality holds. 
    \begin{equation} \label{theorem_3_eq_3} 
        (-xq^l \text{ mod }n)(-yq^l\text{ mod }n) \geq \delta \quad \forall\; (x,y) \in Z_{des},\; l \in \mathbb{Z}_m.
    \end{equation}
    
    Observe that all the points $(x,y)\in Z_{des}$ satisfy $xy<\delta <n$. 
    Therefore, the $q$-ary expansions of $x$ and $y$ must satisfy the following constraint:
    If $q$-ary coefficients $x_i$ of $x$ are equal to zero for $i > k$ and $x_k \neq 0$, then $q$-ary coefficients $y_j$ of $y$ must be equal to zero for $j > m-1-k$. 
    Therefore all the points in $Z_{des}$ must be of the following form, for some $k\in \mathbb{Z}_m$.
    \begin{equation}\label{theorem_3_eq_4} 
        x = \sum_{i=0}^{k}x_i q^i,\ x_k\neq 0\quad \quad y = \sum_{j=0}^{m-1-k}y_j q^j
    \end{equation}
    Based on Eq.\eqref{theorem_3_eq_4}, let us partition the points in $Z_{des}$ into disjoint sets $Z_{des, 0}\cup Z_{des, 1}...\cup Z_{des, m-1}$ as follows.
    \begin{equation} \label{theorem_3_eq_5} 
        Z_{des,k} = \left\{(x,y)\in Z_{des}\ \middle|\ q^k-1 < x \leq q^{k+1}-1 \right\}
    \end{equation}
    Observe that, based on the definition of $Z_{des}$, for every point $(x,y) \in Z_{des}$ $(y,x)$ also belongs to $Z_{des}$. Therefore, it suffices to show the inequality
    in Eq.~\eqref{theorem_3_eq_3} holds for the points in $Z_{des}$ where $x \leq y$. This implies, from Eq.~\eqref{theorem_3_eq_4}~and~\eqref{theorem_3_eq_5}, instead of considering all points in $Z_{des}$, it is enough to consider just the following points: 
    \begin{equation}\label{theorem_3_eq_6} 
        (x,y) \in Z_{des,k} \text{ where } k \leq (m-1)/2.
    \end{equation}
    Now let us try to lower bound the value of $(-xq^l \text{ mod }n)(-yq^l\text{ mod }n)$ for $(x,y) \in Z_{des,k}$, $k \leq (m-1)/2$ and $l \in \{0,1,...,m-1\}$.
    \begin{compactenum}[i)]
        \item When $m$ is even: $\delta < q^m-1-q^{\floor{m/2}}$. From Lemma~\ref{apx_lemma_1}, 
        we have the following equality.
        \begin{align}\label{theorem_3_eq_7} 
            \min\limits_{\substack{(x,y) \in Z_{des}\\ 0\leq l\leq m-1}} (-xq^l\text{ mod }n)(-yq^l\text{ mod }n) &\geq \min\left\{q^m-1-q^{\ceil{m/2}-1},(q^{m/2}-1)^2\right\} \nonumber\\
            &= (q^{m/2}-1)^2 = \delta\ (\text{when }m\text{ is even}).
        \end{align}
        
        \item When $m$ is odd: $\delta = q^m-1-q^{\floor{m/2}}$. From Lemma~\ref{apx_lemma_1}, it follows that
        \begin{align}\label{theorem_3_eq_8} 
            \min\limits_{\substack{(x,y)\in Z_{des}\\ 0\leq l \leq m-1}} (-xq^l\bmod n)(-yq^l\bmod n) &= q^m-1-q^{\floor{m/2}} \nonumber\\
            &= \delta\ (\text{when }m\text{ is odd}).
        \end{align}
    \end{compactenum}
    This shows that every point $(u,v)\in Z^{-1}$ satisfies $uv\geq \delta$ and therefore cannot be in 
    $Z_{des}$ and thus satisfying Eq.~\eqref{lemma_2_eq_1}. Next we show that $d\leq \delta$ is a necessary condition for $C=\mathscr{H}(n\times n,q;d)$ to be dual containing. 
    Seeking a contradiction let us assume that $ \mathscr{H}(n\times n,q;d)^{\perp} \subseteq \mathscr{H}(n\times n,q;d)$ for some $d > \delta$.
    \begin{enumerate}
        \item When $m$ is even: consider the point $(x,y) = (\sqrt{\delta},\sqrt{\delta})$. Since $\delta < d$, from the definitions of $Z_{des}$ \eqref{theorem_3_eq_2} and $Z^{-1}$,
        \begin{equation}\label{theore3_eq_8_1}
            (\sqrt{\delta},\sqrt{\delta}) \in Z_{des}
            \quad \&\quad 
            \left((-\sqrt{\delta}q^l\text{ mod }n),(-\sqrt{\delta}q^l\text{ mod }n)\right) \in Z^{-1}.
        \end{equation}
        Since $\delta = (q^{m/2}-1)^2$, for $l=m/2$ we have the following equality.
        \begin{align}\label{theorem_3_eq_9}
            (-\sqrt{\delta}q^l) \bmod n &= -q^{m}+q^{m/2} \bmod n \nonumber\\
            &= q^{m/2}-1=\sqrt{\delta}
        \end{align}
        Therefore from \eqref{theore3_eq_8_1} and \eqref{theorem_3_eq_9}, we can conclude that $(\sqrt{\delta},\sqrt{\delta}) \in Z_{des}\cap Z^{-1}$.
        \item When $m$ is odd: consider the point $(\delta,1)$. Since $\delta < d$, 
        \begin{equation}
            (\delta,1) \in Z_{des}
            \quad \&\quad
            \left((-\delta q^l\bmod n),(-q^l\bmod n)\right) \in Z^{-1}
        \end{equation}
        If $l = (m+1)/2$, then we have 
        \begin{subequations}\label{theorem_3_eq_10} 
            \begin{align}
                -\delta q^{(m+1)/2} \bmod  n
                &= -(q^m-1-q^{(m-1)/2})q^{(m+1)/2} \bmod n = 1 \\
                -q^{(m+1)/2}\bmod  n 
                &= (q^m-1-q^{(m+1)/2}) < \delta < d
            \end{align}
        \end{subequations}
                Hence, $(\delta,1) \in Z_{des}\cap Z^{-1}$.
    \end{enumerate}
    In both the above cases, we have $Z_{des}\cap Z^{-1} \neq \emptyset$. This implies, by Lemma~\ref{lemma_2}, $C$ cannot contain $C^\perp$ giving us the desired contradiction.
\end{proof}
Now let see the structural results on the cosets in two dimensions which are used in the above theorem \ref{theorem_3}.

\begin{lemma} \label{apx_lemma_1}
Suppose $n=q^m-1$, $m>3$ and the sets $Z_{des}$ and $Z_{des,k}$ are as follows where $0 \leq k \leq (m-1)/2$. 
\begin{align}
        Z_{des} = \left\{(x,y)\ \middle|\ xy < q^m-1-q^{\floor{m/2}},1\leq x,y < n \right\}\label{apx_lemma_1_eq_1}\\
        Z_{des,k} = \left\{(x,y)\ \middle|\ q^k-1 < x \leq q^{k+1}-1, (x,y) \in Z_{des} \right\}\label{apx_lemma_1_eq_2}
\end{align}

Let $f(x,y,l) := (-xq^l\bmod n)(-yq^l\bmod  n)$, where $ l \in \mathbb{Z}_m$.
\begin{subequations}
\begin{align}
\min\limits_{\substack{(x,y)\in Z_{des,k}\\ l \neq m-k-1}} f(x,y,l) &=  q^m-1-q^{\ceil{m/2}-1}; \quad \text{occurs when } l=\ceil{m/2}-1\label{eq:neq-m-k-1}\\
\min\limits_{\substack{(x,y)\in Z_{des,k}\\ l = m-k-1}} f(x,y,l) &= \begin{cases}
(q^{\floor{m/2}}-1)^2 ;   \quad \text{occurs when }m \text{ is even and } l = m/2\\
\geq \quad q^m-1 ; \quad \text{otherwise}
\end{cases}\label{eq:eq-m-k-1}
\end{align}
\end{subequations}
\end{lemma}

\begin{proof}
Let $(x,y) \in Z_{des,k}$. By equation \eqref{apx_lemma_1_eq_1}, we have 
$xy < q^m-1$, and therefore $x,y$ should have the following $q$-ary expansions.
\begin{equation}\label{apx_lemma_1_eq_5}
    x = \sum\limits_{i=0}^{k} x_i q^i,\ x_k \neq 0 \quad\quad y = \sum\limits_{j=0}^{m-k-1} y_j q^j
\end{equation}
where $0 \leq x_i, y_j \leq q-1$. Note that $y_{m-k-1}$ need not be nonzero here. Correspondingly $(n-x,n-y)$ will have the following $q$-ary form.
\begin{equation}\label{apx_lemma_1_eq_6}
    (n-x) = \sum\limits_{i=0}^{k} (q-1-x_i)q^i + \sum_{i=k+1}^{m-1} (q-1)q^i,\ x_k\neq 0
    \quad \& \quad
    (n-y) = \sum\limits_{j=0}^{m-k-1} (q-1-y_j)q^i + \sum\limits_{j=m-k}^{m-1} (q-1)q^j
\end{equation}
The $q$-ary expansions of $(-xq^l\text{ mod }n)$ and $(-yq^l\text{ mod }n)$ are obtained by taking the $l^{th}$ right circular shift of $q$-ary expansions of $(n-x)$ and $(n-y)$ respectively.

\begin{compactenum}[i)]
    \item \label{apx_lemma_1_case_1} ${0 \leq l < m-k-1:}$
    \begin{equation}\label{apx_lemma_1_eq_7}
        (-xq^l\text{ mod } n) = \left(\sum\limits_{i=0}^{l-1} (q-1)q^i\right) + q^l\left(\sum\limits_{i=0}^{k}(q-1-x_i)q^i\right) + \left(\sum\limits_{i=k+l+1}^{m-1} (q-1)q^i\right)
    \end{equation}    
    When $l<m-k-1$, $(m-1)^{th}$ $q$-ary coefficient of $\left(-xq^l\text{ mod }n\right)$ is equal to $(q-1)$. 
    Hence  $\left(-xq^l\text{ mod }n\right)\geq(q-1)q^{m-1}$
    If $(-yq^l \bmod n)>1$, then 
    \begin{align}
    \left(-xq^l\text{ mod }n\right) (-yq^l \bmod n) &\geq 2(q-1)q^{m-1} \nonumber\\
    &=2(q^m-q^{m-1}) = q^m+ (q^m-2q^{m-1})>q^m-1 
    \end{align}
    If the minimum value of $\left(-xq^l\text{ mod }n\right)\left(-yq^l\text{ mod }n\right)$ is less than $q^m-1$,
    then $\left(-yq^l\text{ mod }n\right)$ must be equal to $1$,  equivalently, $y$ must be equal to $q^m-1-q^{m-l}$.
    Given that $xy < q^m-1-q^{\floor{m/2}}$, $y = q^m-1-q^{m-l}$ implies $l \leq \ceil{m/2}-1$ and $x=1$.
    Under these conditions minimum value of $\left(-xq^l\text{ mod }n\right)\left(-yq^l\text{ mod }n\right)$ occurs when $l = \ceil{m/2}-1$ which is equal to $q^m-1-q^{\ceil{m/2}-1}$.
    \item \label{apx_lemma_1_case_2} $l = m-k-1:$
    \begin{subequations}\label{apx_lemma_1_8}
    \begin{align}
        (-xq^{l}\text{ mod } n) &= \left(\sum\limits_{i=0}^{m-k-2} (q-1)q^i\right) + q^{m-k-1}\left(\sum\limits_{i=0}^{k} (q-1-x_i)q^i\right)\\
        (-yq^{l}\text{ mod } n) &= \left(\sum\limits_{j=0}^{m-2k-2} (q-1-y_{j+k+1})q^j\right) + \left(\sum\limits_{j=m-2k-1}^{m-k-2}(q-1)q^j\right) + q^{m-k-1}\left(\sum\limits_{j=0}^{k} (q-1-y_j)q^j\right)
    \end{align}
    \end{subequations}
    Let us assume that the minimum value of $\left(-xq^l\text{ mod }n\right)\left(-yq^l\text{ mod }n\right)$ is less than $q^m-1$. Let us check if our assumption is a valid one and if so, we will find $(x,y)$ which satisfy our assumption.
    If $k=0$ then $\left(-xq^l\text{ mod }n\right)$ is greater than or equal to $q^{m-1}-1$.
    From Eq.~\eqref{apx_lemma_1_eq_1}, $y_j<(q-1)$ for some $\floor{m/2}\leq j\leq m-1$. Hence $(-yq^l\text{ mod }n)$ is greater than or equal to $q^{\floor{m/2}-1}$. 
    Therefore, the product $\left(-xq^l\text{ mod }n\right)\left(-yq^l\text{ mod }n\right)$ is greater than $q^m-1$. When $k\neq 0$, $\left(-xq^l\text{ mod }n\right)$ is greater than or equal to $q^{m-k-1}-1$
    and $\left(-yq^l\text{ mod }n\right)$ is greater than or equal to $q^{m-2k-1}(q^k-1)$ which implies the product $\left(-xq^l\text{ mod }n\right)\left(-yq^l\text{ mod }n\right)$ is greater than or equal to $(q^{m-k-1}-1)(q^k-1)q^{m-2k-1}$. Combining this with $\left(-xq^l\text{ mod }n\right)\left(-yq^l\text{ mod }n\right) < q^m-1$, we get $k \geq m/2-1$. we also have the inequality $k \leq (m-1)/2$. This implies the possible values of $k$ are $(m-1)/2$ when $m$ is odd and $m/2 -1$ when $m$ is even.
    \begin{compactenum}[(a)]
        \item When ${k = (m-1)/2}$,
        \begin{subequations}\label{apx_lemma_1_eq_10}
        \begin{align}
            (-xq^{l}\text{ mod } n) &= \left(\sum\limits_{i=0}^{(m-3)/2} (q-1)q^i\right) + q^{(m-1)/2}\left(\sum\limits_{i=0}^{(m-1)/2} (q-1-x_i)q^i\right)\\
            (-yq^{l}\text{ mod } n) &= \left(\sum\limits_{j=0}^{(m-3)/2}(q-1)q^j\right) + q^{(m-1)/2}\left(\sum\limits_{j=0}^{(m-1)/2} (q-1-y_j)q^j\right)
        \end{align}
        \end{subequations}
        Since $xy$ is less than $(q^m-1-q^{\floor{m/2}})$, at least on of the $\left\{x_{(m-3)/2},\ y_{(m-3)/2},\ y_{(m-1)/2}\right\}$ must be equal to zero. Otherwise $xy$ will be greater than $q^m-1$.
        This implies either $(-xq^l\text{ mod }n)$ or $(-yq^l\text{ mod }n)$ must be greater than or equal to $q^{(m-1)/2}-1+q^{m-2}$ and the product $\left\{(-xq^l\text{ mod }n)(-yq^l\text{ mod }n)\right\} \geq (q^{(m-1)/2}-1)(q^{(m-1)/2}-1+q^{m-2})$ cannot be less than $q^m-1$. 
        
        \item When ${k = m/2 -1}$,
        \begin{subequations}
        \begin{align}
            (-xq^{l}\text{ mod } n) &= \left(\sum\limits_{i=0}^{m/2-1} (q-1)q^i\right) + q^{m/2}\left(\sum\limits_{i=0}^{m/2-1} (q-1-x_i)q^i\right)\\
            (-yq^{l}\text{ mod } n) &= \left(q-1-y_{m/2}\right) + \left(\sum\limits_{j=1}^{m/2-1}(q-1)q^j\right) + q^{m/2}\left(\sum\limits_{j=0}^{m/2-1} (q-1-y_j)q^j\right)
        \end{align}
        \end{subequations}
        Since $xy$ is less than $(q^m-1-q^{\floor{m/2}})$, at least on of the $\left\{x_{m/2-2},\ x_{m/2-1},\ y_{m/2},\ y_{m/2-1}\right\}$ must be equal to zero. Therefore, from the above equations, minimum value of  $\left(-xq^l\text{ mod }n\right)\left(-yq^l\text{ mod }n\right)$ occurs when $(y_{m/2}) = 0$, and the minimum value is equal to $(q^{m/2}-1)^2$
    \end{compactenum}
    
    \item \label{apx_lemma_1_case_3} $m-k-1 < l \leq m-1:$
    \begin{subequations}\label{apx_lemma_1_eq_12}
    \begin{align}
        (-xq^{l}\text{ mod } n) &= \left(\sum\limits_{i=0}^{l-(m-k)} (q-1-x_{i+m-l})q^i\right) + \left(\sum\limits_{i=l-(m-k-1)}^{l-1} (q-1)q^i\right) + q^l\left(\sum\limits_{i=0}^{m-l-1} (q-1-x_i)q^i\right)\\
        (-yq^{l}\text{ mod } n) &= \left(\sum\limits_{j=0}^{l-k-1} (q-1-y_{j+m-l})q^j\right) + \left(\sum\limits_{j=l-k}^{l-1}(q-1)q^j\right) + q^l\left(\sum\limits_{j=0}^{m-l-1} (q-1-y_j)q^j\right)
    \end{align}
    \end{subequations}
   
   Assume that the minimum value of $\left(-xq^l\text{ mod }n\right)\left(-yq^l\text{ mod }n\right)$ is less than $q^m-1$. Let us check if our assumption is a valid one. From equation \eqref{apx_lemma_1_eq_12}, we can say that $(-xq^l\text{ mod }n) \geq q^{l-m+k+1}(q^{m-k-1}-1)$ and $(-yq^l\text{ mod }n) \geq q^{l-k}(q^k-1)$. This, together with our assumption, implies $2l-2$ must be less than $m$. 
   Additionally, the inequalities $m-k-1 < l$ and $k \leq (m-1)/2$ imply
   $l = (m+1)/2$ and $k = (m-1)/2$. Subsequently the following inequality holds true.
    \begin{align}\label{apx_lemma_1_eq_13}
        (-xq^l\text{ mod }n)(-yq^l\text{ mod }n) &\geq q^2\left(q^{(m-1)/2}-1\right)^2 \nonumber\\
        &> q^m-1
    \end{align}    
    Therefore our assumption is not a valid one and the minimum value of $\left(-xq^l\text{ mod }n\right)\left(-yq^l\text{ mod }n\right)$ cannot be less than $q^m-1$.
\end{compactenum}
The cases \ref{apx_lemma_1_case_1} and \ref{apx_lemma_1_case_3} imply Eq. \eqref{eq:neq-m-k-1} and the case \ref{apx_lemma_1_case_2} implies Eq. \eqref{eq:eq-m-k-1}.
\end{proof}

We can extend theorem \ref{theorem_3} to the case of non-primitive narrow-sense bicyclic hyperbolic codes. The following corollary gives a sufficiency condition to verify if a narrow-sense bicyclic hyperbolic code contains its Euclidean dual code.

\begin{corollary}\label{theorem_7}
Suppose $m = ord_n(q)$. Narrow-sense bicyclic hyperbolic code of length $n\times n$ over $\mathbb{F}_{q^2}$ contains its Euclidean dual if the design distance $d$ satisfies $2 \leq d \leq \Delta$, where
\begin{equation}
        \Delta = 
    \begin{cases}
        \Big(\frac{n^2}{(q^m-1)^2}\Big)\Big[(q^m-1) - 2(q^{m/2} - 1)\Big] & m \text{ is even} \vspace{10pt}\\
        \Big(\frac{n^2}{(q^m-1)^2}\Big)\Big[(q^m-1) - (q^{\frac{m-1}{2}})\Big] & m \text{ is odd}
    \end{cases}
\end{equation}
\end{corollary}

\begin{proof}
Let $C= \mathscr{H}(n\times n, q;d)$ be a narrow-sense bicyclic hyperbolic code with designed distance $d$. For proving $C \supseteq C^\perp$ when $d \leq \Delta$, it is  enough to show that
$D=\mathscr{H}(n\times n, q;\Delta)$ contains $D^\perp$.
Since $D\subseteq C$, this implies $C^\perp \subseteq C$. Let $\overline{Z}$ be the defining set of $D$. From lemma \ref{lemma_2}, our goal is to prove that
\begin{equation}\label{theorem_7_eq_1}
    (-\overline{x}q^l\text{ mod }n)(-\overline{y}q^l\text{ mod }n) \geq \Delta \quad \quad \forall\ (\overline{x},\overline{y})\in \overline{Z}_{des},\ l\in\{0,1,...,m-1\}
\end{equation}
Let $(\overline{x},\overline{y}) \in \overline{Z}_{des}$. Since $0\leq \overline{x}<n$ and $0\leq \overline{y}<n$, they can be written in the following form
\begin{subequations}\label{theorem_7_eq_2}
\begin{align}
    \overline{x} &= \frac{n}{q^m-1}\Big(\frac{q^m-1}{n}(\overline{x})\Big) = \frac{n}{q^m-1}(x) \\
    \overline{y} &= \frac{n}{q^m-1}\Big(\frac{q^m-1}{n}(\overline{y})\Big) = \frac{n}{q^m-1}(y) \numberthis
\end{align}
\end{subequations}
where $0\leq x < q^m-1$,  $0\leq y < q^m-1$. Correspondingly
\begin{subequations}
\begin{align}
    (-\overline{x}q^l\text{ mod }n) &= \left(\frac{n}{q^m-1}\right)(-xq^l\text{ mod }(q^m-1)) \\
    (-\overline{y}q^l\text{ mod }n) &= \left(\frac{n}{q^m-1}\right)(-yq^l\text{ mod }(q^m-1))
\end{align}
\end{subequations}
From above equations, we can say that $\overline{x}\overline{y} < \Delta$ implies $xy < \delta$, where $\delta$ is from theorem \ref{theorem_3} and $(-xq^l\text{ mod }(q^m-1))(yq^l\text{ mod }(q^m-1)) \geq \delta$ implies  $(-\overline{x}q^l\text{ mod }n)(-\overline{y}q^l\text{ mod }n) \geq \Delta$. Therefore from theorem \ref{theorem_3}, we can say that equation \eqref{theorem_7_eq_1} is true.
\end{proof}

Now that we have a necessary and sufficient condition for narrow-sense primitive bicyclic hyperbolic code and a sufficiency condition for non-primitive narrow-sense bicyclic hyperbolic code to contain their corresponding Euclidean dual codes, the CSS construction enables us to construct quantum stabilizer codes.

\begin{proposition}[Calderbank-Shor-Steane (CSS) construction,\cite{calderbank98}]\label{prop:css}
If there exists an $[n,k,d]$ Euclidean dual containing classical linear code $C$ over $\mathbb{F}_q$, then there exists an $[[n,2k-n,d]]$ stabilizer code over $\mathbb{F}_q$.
\end{proposition}

\begin{corollary}[Quantum bicylic codes I]
Let $q$ be a prime power and $n=q^m-1$  for $m>3$ and $d< \delta$ as in Theorem~\ref{theorem_3}.
Then there exists a quantum bicyclic code of length $n^2$ and distance $\geq d$.
\end{corollary}

\section{Hermitian Dual Containing Codes}\label{sec:main-hd}

Let $u\in\mathbb{F}_{q^2}^{n_1\times n_2}$, we define  $u^q = (u_{ij}^q)$.
Suppose $C$ is a linear code of size $n_1\times n_2$ over $\mathbb{F}_{q^2}$. 
Then $C^{\perp_h}$, the Hermitian dual code  of $C$ is defined as 
\begin{align}
C^{\perp_h} = \left\{u\in \mathbb{F}^{n_1\times n_2}_{q}\ \middle|\ u^{q}\cdot c = 0\  \mbox{ for all }\  c \in C\right\}.\label{eq:h-dual}
\end{align} 
From Eq.~\eqref{eq:h-dual}, we can see that the following relations hold between $C^\perp$ and $C^{\perp_h}$.
\begin{subequations}\label{eq:dual-reln}
\begin{align}
C^{\perp} =&\left\{u^q\ \middle|\ u\in C^{\perp_h}  \right\} \label{eq:dual-reln-h2e}\\
C^{\perp_h} =&\left\{u^q\ \middle|\ u\in C^\perp  \right\} \label{eq:dual-reln-e2h}
\end{align}
\end{subequations}
If $c(X,Y)$ is a code polynomial in $C^\perp$, then 
$c(X,Y)^q$ is a code polynomial in $C^{\perp_h}$.
This implies that if $(x,y)$ is a common zero of $C^\perp$,
then $(x^{q},y^{q})$ is a common zero of $C^{\perp_h}$.
Suppose $Z$ is the defining set of bicyclic code $C$ of length $n_1\times n_2$ over the field $\mathbb{F}_{q^2}$ and gcd$(n_1n_2,q)=1$, then Eqs.~\eqref{eq:z-dual}~and~\eqref{eq:dual-reln} imply that the defining set of $C^{\perp_h}$ is
\begin{align}
Z^{\perp_h}=Z_{tot} \setminus Z^{-q}, \label{eq:z-hdual}
\end{align}
where $Z^{-q} = \left\{ (-xq \bmod n_1,-yq \bmod n_2) \middle|\ (x,y)\in Z\right\}$. 
Analogous to Lemma~\ref{lemma_2}, the following lemma gives a simple condition to check if a bicyclic code contains its Hermitian dual code. 


\begin{lemma}\label{lemma_4}
Let $\text{gcd}(n_1n_2,q) = 1$ and $C$ be a bicyclic code of length $n_1\times n_2$ over $\mathbb{F}_{q^2}$. Let $Z$ be the defining set of $C$ and $Z_{des}$
the designed set. Then $ C^{\perp_h} \subseteq C$ 
if and only if 
\begin{equation}\label{lemma_4_eq_1}
    Z_{des}\cap Z^{-q} = \emptyset
\end{equation}
where $Z^{-q} = \left\{ (-qx \bmod n_1, -qy \bmod n_2)\middle|\ (x,y)\in Z\right\}$.
\end{lemma}

\begin{proof}
If $Z$ is the defining set of $C$, then the defining set of $C^{\perp_h}$ is $Z_{tot} \setminus Z^{-q}$. 
Therefore, $C^{\perp_h}$ is contained in $C$ if and only if the defining set of $C$ is contained in defining  set of $C^{\perp_h}$.
\begin{equation}\label{lemma_4_eq_2}
    Z \subseteq Z_{tot}\setminus Z^{-q} \iff Z \cap Z^{-q} = \emptyset
\end{equation}
Since $Z$ is a union of cyclotomic cosets, 
the condition $Z\cap Z^{-q} = \emptyset$ is true if and only if there is no common coset between $Z$ and $Z^{-q}$. By Eq.~\eqref{eq:bch-Z} every 
cyclotomic coset in $Z$ contains at least one element from $Z_{des}$. Therefore, $Z\cap Z^{-q} = \emptyset$ is true if and only if $Z_{des}\cap Z^{-q} = \emptyset$.
\end{proof}

Extending lemma \ref{apx_lemma_1}, the following lemma about the structural results on cyclotomic cosets will be useful in proving results for Hermitian dual containing codes. The proof is very much similar to the proof of lemma \ref{apx_lemma_1} with the following key differences.
\begin{itemize}
    \item Range of $x$ and $y$ will be from $0$ to $q^{2m}-1$ and therefore there will be $2m$ $q$-ary coefficients for $x$ and $y$ compared to $m$ $q$-ary coefficients in Euclidean case.
    \item Since the code is over field $\mathbb{F}_{q^2}$, we need to consider $q^2$-ary cyclotomic coset instead of $q$-ary cyclotomic coset. $q^2$-ary cyclotomic coset of $(x,y)$ is $\left\{(xq^{l},yq^{l})\ \middle|\ l \in\left\{2,4...2m\right\}\right\}$.
\end{itemize}
\begin{lemma} \label{apx_corollary_2}
Suppose $n=q^{2m}-1$, $m>3$ and $Z_{des}$ and $Z_{des,k}$ as follows,
where $k \leq (2m-1)/2$
\begin{align}
        Z_{des} = \left\{(x,y)\ \middle|\ xy < q^{2m}-1-q^{m-1},1\leq x,y \leq n \right\}\label{apx_corollary_2_eq_1} \\
        Z_{des,k} = \left\{(x,y)\ \middle|\ q^k-1 < x \leq q^{k+1}-1, (x,y) \in Z_{des} \right\}\label{apx_corollary_2_eq_2}
\end{align}
Define $f(x,y,l):=(-xq^l  \bmod n) (-yq^l ,  \bmod n)$ where
$l \in \{1,3,5...2m-1\}$. Then
\begin{align}
        \min\limits_{\substack{(x,y)\in Z_{des,k}\\ l \neq 2m-k-1}} f(x,y,l) &=  
        \begin{cases}
        q^{2m}-1-q^{m-1}; \quad \text{occurs when } m \text{ is even and } l = m-1\\
        q^{2m}-1-q^{m}; \quad \text{occurs when } m \text{ is odd and } l = m\\
        \end{cases} \label{apx_corollary_2_eq_3}\\
            \min\limits_{\substack{(x,y)\in Z_{des,k}\\ l = 2m-k-1}} f(x,y,l)&=
            \begin{cases}
             (q^{m}-1)^2; \quad \text{occurs when }m \text{ is odd and } l=m\\
             \geq q^{2m}-1; \quad \text{ otherwise}
            \end{cases}   \label{apx_corollary_2_eq_4}
\end{align}
\end{lemma}

\begin{proof}[Proof of Lemma~\ref{apx_corollary_2}]
Let $(x,y) \in Z_{des,k}$. By Eq.~\eqref{apx_corollary_2_eq_1}, we also have 
$xy < q^{2m}-1$, and $(x,y)$ should have the following $q$-ary expansions.
\begin{equation}
    x = \sum\limits_{i=0}^{k} x_i q^i,\ x_k \neq 0 \quad\quad y = \sum\limits_{j=0}^{2m-k-1} y_j q^j
\end{equation}
where $0 \leq x_i, y_j \leq q-1$. Note that $y_{m-k-1}$ need not be nonzero here. Correspondingly $(n-x,n-y)$ will have the following $q$-ary form.
\begin{equation}
    (n-x) = \sum\limits_{i=0}^{k} (q-1-x_i)q^i + \sum_{i=k+1}^{2m-1} (q-1)q^i,\ x_k\neq 0
    \quad \& \quad
    (n-y) = \sum\limits_{j=0}^{2m-k-1} (q-1-y_j)q^i + \sum\limits_{j=2m-k}^{2m-1} (q-1)q^j
\end{equation}
The $q$-ary expansions of $(-xq^l\text{ mod }n)$ and $(-yq^l\text{ mod }n)$ are obtained by taking the $l^{th}$ right circular shift of $q$-ary expansions of $(n-x)$ and $(n-y)$ respectively.

\begin{compactenum}[i)]
    \item ${0 \leq l < 2m-k-1:}$
    \begin{equation}
        (-xq^l\text{ mod } n) = \left(\sum\limits_{i=0}^{l-1} (q-1)q^i\right) + q^l\left(\sum\limits_{i=0}^{k}(q-1-x_i)q^i\right) + \left(\sum\limits_{i=k+l+1}^{2m-1} (q-1)q^i\right)
    \end{equation}    
    When $l<2m-k-1$, $(2m-1)^{th}$ $q$-ary coefficient of $\left(-xq^l\text{ mod }n\right)$ is equal to $(q-1)$. 
    Hence  $\left(-xq^l\text{ mod }n\right)\geq(q-1)q^{2m-1}$
    and if $(-yq^l \bmod n)>1$ then 
    \begin{align}
    \left(-xq^l\text{ mod }n\right) (-yq^l \bmod n) &\geq 2(q-1)q^{2m-1} \\
    &=2(q^{2m}-q^{2m-1}) = q^{2m}+ (q^{2m}-2q^{2m-1})>q^{2m}-1 
    \end{align}
    If the minimum value of $\left(-xq^l\text{ mod }n\right)\left(-yq^l\text{ mod }n\right)$ is less than $q^{2m}-1$,
    then $\left(-yq^l\text{ mod }n\right)$ must be equal to $1$,  equivalently, $y$ must be equal to $q^{2m}-1-q^{2m-l}$.
    Given that $xy < q^{2m}-1-q^{m-1}$, $y = q^{2m}-1-q^{2m-l}$ implies $l \leq m$ and $x=1$.
    Under these conditions minimum value of $\left(-xq^l\text{ mod }n\right)\left(-yq^l\text{ mod }n\right)$ occurs when $l = m$ which is equal to $q^{2m}-1-q^{m}$ if $m$ is odd or when $l=m-1$ which is equal to $q^{2m}-1-q^{m-1}$ if $m$ is even.
    \item $l = 2m-k-1:$
    \begin{subequations}
    \begin{align}
        (-xq^{l}\text{ mod } n) &= \left(\sum\limits_{i=0}^{2m-k-2} (q-1)q^i\right) + q^{2m-k-1}\left(\sum\limits_{i=0}^{k} (q-1-x_i)q^i\right)
        \label{eq:-xql}\\
        (-yq^{l}\text{ mod } n) &= \left(\sum\limits_{j=0}^{2m-2k-2} (q-1-y_{j+k+1})q^j\right) + \left(\sum\limits_{j=2m-2k-1}^{2m-k-2}(q-1)q^j\right) + q^{2m-k-1}\left(\sum\limits_{j=0}^{k} (q-1-y_j)q^j\right)\label{eq:-yql}
    \end{align}
    \end{subequations}
    Let us assume that the minimum value of $\left(-xq^l\text{ mod }n\right)\left(-yq^l\text{ mod }n\right)$ is less than $q^{2m}-1$. Let us check if our assumption is a valid one and if so, we will find $(x,y)$ which satisfy our assumption.
    If $k=0$, then by Eq.~\eqref{eq:-xql}, $\left(-xq^l\text{ mod }n\right)$ is greater than or equal to $q^{2m-1}-1$.
    From Eq.~\eqref{apx_corollary_2_eq_1}, $y_j<(q-1)$ for some $m-1\leq j\leq 2m-1$. Hence $(-yq^l\text{ mod }n)$ is greater than or equal to $q^{m-2}$.
    Therefore, the product $\left(-xq^l\text{ mod }n\right)\left(-yq^l\text{ mod }n\right)$ is greater than $q^{2m}-1$. When $k\neq 0$, $\left(-xq^l\text{ mod }n\right)$ is greater than or equal to $q^{2m-k-1}-1$
    and $\left(-yq^l\text{ mod }n\right)$ is greater than or equal to $q^{2m-2k-1}(q^k-1)$ which implies the product $\left(-xq^l\text{ mod }n\right)\left(-yq^l\text{ mod }n\right)$ is greater than or equal to $(q^{2m-k-1}-1)(q^k-1)q^{2m-2k-1}$. Combining this with $\left(-xq^l\text{ mod }n\right)\left(-yq^l\text{ mod }n\right) < q^{2m}-1$, we get $k \geq m-1$. we also have the inequality $k \leq (2m-1)/2$. Combining both the inequalities we get $k = m-1$.
    When ${k = m -1}$,
    \begin{subequations}
    \begin{align}
        (-xq^{l}\text{ mod } n) &= \left(\sum\limits_{i=0}^{m-1} (q-1)q^i\right) + q^{m}\left(\sum\limits_{i=0}^{m-1} (q-1-x_i)q^i\right)\\
        (-yq^{l}\text{ mod } n) &= \left(q-1-y_{m}\right) + \left(\sum\limits_{j=1}^{m-1}(q-1)q^j\right) + q^{m}\left(\sum\limits_{j=0}^{m-1} (q-1-y_j)q^j\right)
    \end{align}
    \end{subequations}
    Since $xy$ is less than $(q^{2m}-1-q^{m-1})$, at least on of the $\left\{x_{m-2},\ x_{m-1},\ y_{m},\ y_{m-1}\right\}$ must be equal to zero. Therefore, from the above equations, minimum value of  $\left(-xq^l\text{ mod }n\right)\left(-yq^l\text{ mod }n\right)$ occurs when $(y_{m}) = 0$, and the minimum value is equal to $(q^{m}-1)^2$. Note that since $l$ is restricted to only odd values, the minimum value $(q^m -1)^2$ occurs only when $m$ is odd. Otherwise the minimum value is greater than $q^{2m}-1$. 
    
    \item  $2m-k-1 < l \leq 2m-1:$
    \begin{subequations}\label{apx_corollary_2_eq_12}
    \begin{align}
        (-xq^{l}\text{ mod } n) &= \left(\sum\limits_{i=0}^{l-(2m-k)} (q-1-x_{i+2m-l})q^i\right) + \left(\sum\limits_{i=l-(2m-k-1)}^{l-1} (q-1)q^i\right) + q^l\left(\sum\limits_{i=0}^{2m-l-1} (q-1-x_i)q^i\right)\\
        (-yq^{l}\text{ mod } n) &= \left(\sum\limits_{j=0}^{l-k-1} (q-1-y_{j+2m-l})q^j\right) + \left(\sum\limits_{j=l-k}^{l-1}(q-1)q^j\right) + q^l\left(\sum\limits_{j=0}^{2m-l-1} (q-1-y_j)q^j\right)
    \end{align}
    \end{subequations}
   
   Assume that the minimum value of $\left(-xq^l\text{ mod }n\right)\left(-yq^l\text{ mod }n\right)$ is less than $q^{2m}-1$. 
   Let us check if our assumption is a valid one. From equation \eqref{apx_corollary_2_eq_12}, we can say that $(-xq^l\text{ mod }n) \geq q^{l-2m+k+1}(q^{2m-k-1}-1)$ and $(-yq^l\text{ mod }n) \geq q^{l-k}(q^k-1)$. 
   This, together with our assumption, implies $2l-2$ must be less than $2m$.
   Observe that there doesn't exist $l,k$ which satisfy $2l-2<2m$ along with the inequalities $2m-k-1 < l$ and $k \leq (2m-1)/2$. 
   Therefore our assumption is not a valid one and the minimum value of $\left(-xq^l\text{ mod }n\right)\left(-yq^l\text{ mod }n\right)$ cannot be less than $q^{2m}-1$.
\end{compactenum}
\end{proof}

The following theorem gives an easy condition based on the designed distance to determine if a primitive narrow-sense hyperbolic code contains its Hermitian dual.

\begin{theorem}\label{theorem_5}
A primitive narrow-sense hyperbolic code of length $n\times n$ over $\mathbb{F}_{q^2}$, where $n=q^{2m}-1$, with $m>3$, contains its Hermitian dual if and only if the design distance $d$ satisfies $2 \leq d \leq \delta_h$, where
\begin{equation}\label{theorem_5_eq_1}
    \delta_h =
    \begin{cases}
        (q^m-1)^2 & m \text{ is odd} \\
        q^{2m} - 1 - q^{m-1} & m \text{ is even}
    \end{cases}
\end{equation}
\end{theorem}
%
\begin{proof}
    For proving $\mathscr{H}(n\times n,q^2;d)^{\perp_h} \subseteq \mathscr{H}(n\times n,q^2;d)$ when $d \leq \delta_h$, it is enough to show that $\mathscr{H}(n\times n,q^2;\delta_h)^{\perp_h} \subseteq \mathscr{H}(n \times n,q^2;\delta_h)$ since $\mathscr{H}(n\times n,q^2;d)$ contains $\mathscr{H}(n\times n,q^2;\delta_h)$. 
    Let $Z_{des}$ and $Z$ be the designed set and defining set of $\mathscr{H}(n\times n,q^2;\delta_h)$. 
    We need to show that for any $(x,y) \in Z_{des}$, $(-qxq^{2j}, -qyq^{2j} ) \bmod n \not\in Z_{des}$   for $j\in \mathbb{Z}_m$.   Alternatively, we need to show that
    \begin{equation}\label{theorem_5_eq_2}
        (-xq^{2j+1} \bmod n)(-yq^{2j+1} \bmod n) \geq \delta_h \mbox{ for }j\in \mathbb{Z}_m
    \end{equation}
    Equivalently, 
    $(-xq^{l} \bmod n)(-yq^{l} \bmod n) \geq \delta_h$ for 
    $l\in \{1, 3, 5, \ldots, 2m-1\}$.
    Similar to the Euclidean case, based on the $q$-ary expansions of $x,y$  as in Eq.~\eqref{theorem_3_eq_4}, we partition the points in $Z_{des}$ into disjoint sets $Z_{des,0}, Z_{des,1}, \ldots, Z_{des,2m-1}$. And for every $(x,y)$ in $Z_{des}$, $(y,x)$ is also in $Z_{des}$. This implies it is enough to  consider the following points instead of $Z_{des}$.
    \begin{equation}\label{theorem_5_eq_3}
        (x,y) \in Z_{des,k}, \text{ where } k \leq (2m-1)/2
    \end{equation}

    \begin{compactenum}[1)]
        \item  When $m$ is even: Since $\delta_h = q^{2m}-1-q^{m-1}$, from Lemma~\ref{apx_corollary_2}, it follows that
    \begin{align}\label{theorem_5_eq_4}
        \min\limits_{\substack{(x,y) \in Z_{des}\\ l \in \{1,3..2m-1\}}} (-xq^l\text{ mod }n)(-yq^l\text{ mod }n) = q^{2m}-1-q^{m-1} = \delta_h
    \end{align}
    \item When $m$ is odd: Since $\delta_h < q^{2m}-1-q^{m-1}$, from Lemma~\ref{apx_corollary_2}, it follows that
    \begin{align}\label{theorem_5_eq_5}
       \min\limits_{\substack{(x,y) \in Z_{des}\\ l \in \{1,3..2m-1\}}} (-xq^l\text{ mod }n)(-yq^l\text{ mod }n) &\geq \min\left\{q^{2m}-1-q^{m},(q^m-1)^2\right\} \nonumber\\
        &=(q^m-1)^2 = \delta_h
    \end{align}
\end{compactenum}

Let $C= \mathscr{H}(n\times n,q^2;d)$ and $Z_{des}$, and $Z$ its designed set
and defining set respectively. 
For proving the necessity of $d\leq \delta_h$ for $C \supseteq C^{\perp_h}$, we show 
that if $d>\delta_h$, then there exists a point $(x,y)\in Z_{des}$ and $l$ such that $(u,v)=(-qxq^{2l}, -qyq^{2l}) \bmod n \in Z^{-q}$ is also in $Z_{des}$.

\begin{enumerate}
    \item  When $m$ is odd: let $(x,y) = (\sqrt{\delta_h},\sqrt{\delta_h})$ and $l=(m-1)/2$. 
    Since $\delta_h < d$, $(\sqrt{\delta_h},\sqrt{\delta_h}) \in Z_{des}$.
    We have 
    \begin{align}
            -q\sqrt{\delta_h}q^{m-1} \bmod n &= -(q^m-1)q^{m} \bmod n \nonumber\\
            &= q^m -1 =\sqrt{\delta_h} 
    \end{align}
    Therefore, $(-qxq^{2l}\bmod n)(-qyq^{2l}\bmod n) = \delta_h <d$. 
    \item When $m$ is even: let $(x,y) = (\delta_h,1)$ and $l = m/2$. Then $(\delta_h,1) \in Z_{des}$.
    \begin{subequations}\label{theorem_5_eq_7}
        \begin{align}
            -qxq^{2l} \bmod n &= -(q^{2m}-1-q^{m-1})q^{m+1} \bmod n\nonumber\\
            &= 1 \\
            -qyq^{2l} \bmod  n
            &= -q^{m+1}\bmod (q^{2m}-1) \nonumber \\
            &= (q^{2m}-1-q^{m+1}) < \delta_h < d
        \end{align}
    \end{subequations}
\end{enumerate}
When $d > \delta_h$, for both even and odd $m>3$, we have  $Z_{des} \cap Z^{-q} \neq \emptyset$.
By Lemma~\ref{lemma_4}, $\mathscr{H}(n \times n,q^2;d)$ cannot contain its Hermitian dual when $d > \delta_h$.
\end{proof}

Similar to the Euclidean dual case, theorem \ref{theorem_5} can be extended to the non-primitive narrow-sense bicyclic hyperbolic codes. The following corollary gives a sufficiency condition to verify if a narrow-sense bicyclic hyperbolic code contains its Euclidean dual code.

\begin{corollary}\label{corollary_12}
Suppose $m = ord_n(q^2)$. Narrow-sense bicyclic hyperbolic code of length $n\times n$ over $\mathbb{F}_{q^2}$ contains its Hermitian dual if the design distance $d$ satisfies $2 \leq d \leq \Delta_h$, where
\begin{equation}
        \Delta_h = 
    \begin{cases}
        \left(\frac{n^2}{(q^{2m}-1)^2}\right)(q^m-1)^2 & m \text{ is even} \vspace{10pt}\\
        \left(\frac{n^2}{(q^{2m}-1)^2}\right)\left(q^{2m}-1 - q^{m-1}\right) & m \text{ is odd}
    \end{cases}
\end{equation}
\end{corollary}

Based on the conditions obtained from theorem \ref{theorem_5} and \ref{corollary_12} for Hermitian dual containing, we can use the bicyclic hyperbolic codes to construct new quantum stabilizer codes.

\begin{proposition}[Hermitian  construction \cite{calderbank98,ashikhmin01}]
Let $C$ be an $[n,k,d]$ over $\mathbb{F}_{q^2}$ such that $C^{\perp_h} \subseteq C$, then there exists an $[[n,2k-n,d]]$ stabilizer code over $\mathbb{F}_q$.
\end{proposition}

\begin{corollary}[Quantum bicyclic codes II]
Let $q$ be a prime power, $n=q^{2m}-1$  for $m>3$ and $d< \delta_h$ as in Theorem~\ref{theorem_5}.
Then there exist quantum bicyclic codes of length $n^2$ and distance $\geq d$.
\end{corollary}

\section{Conclusion}
In this paper, we have studied bicyclic hyperbolic codes and proved some interesting structural properties of these codes.
Using these results, we can construct new quantum bicyclic codes.
There are other interesting structural properties of bicyclic codes worth further investigation.
One natural direction would be to compute the dimension and actual distance of these codes. Another exciting direction would be to check if our result can be generalized and find the existence of dual containing bicyclic hyperbolic codes of length $n_1\times n_2$ with design distance $\mathcal{O}(\sqrt{n_1n_2})$.
The theory of bicyclic codes is very rich, and we hope this work will motivate further research into quantum bicyclic codes.


\bibliographystyle{plain}
\bibliography{references}

\end{document}